\documentclass[12pt]{article}

\usepackage[utf8]{inputenc}
\usepackage{amsmath,amssymb,amsthm}
\usepackage{verbatim,url,framed,graphicx}
\usepackage[usenames,dvipsnames]{xcolor}
\usepackage[a4paper]{geometry}
\usepackage{bm}
\usepackage{enumitem}

\newcommand{\FF}{\mathbb{F}}
\newcommand{\HH}{\mathbf{H}}

\newcommand{\cv}{{\bf c}}
\newcommand{\ev}{{\bf e}}

\newcommand{\sv}{{\bf s}}
\newcommand{\tv}{{\bf t}}
\newcommand{\uv}{{\bf u}}
\newcommand{\vv}{{\bf v}}
\newcommand{\xv}{{\bf x}}
\newcommand{\yv}{{\bf y}}
\newcommand{\rv}{{\bf r}}

\newcommand{\Rc}{\mathcal R}

\newcommand{\bH}{\mathbf H}

\newcommand{\transpose}[1]{#1^{\intercal}}
\def\rank#1{{\rm rank}{\, #1}}
\def\scalprod#1#2{({#1}\, |\, {#2})}

\DeclareMathOperator{\Ker}{Ker}
\DeclareMathOperator{\Id}{\mathbf I}

\newenvironment{dashlist}{\begin{list}{--}{\itemsep=0mm}}{\end{list}}

\newtheorem{theorem}{Theorem}
\newtheorem{defn}[theorem]{Definition}
\newtheorem{lem}[theorem]{Lemma}
\newtheorem{cor}[theorem]{Corollary}
\newtheorem{prop}[theorem]{Proposition}

\title{Coding Constructions for Efficient Oblivious Transfer from Noisy Channels}
\author{Frédérique Oggier 
\thanks{
Division
of Mathematical Sciences, School of Physical and Mathematical
Sciences, Nanyang Technological University, Singapore.
Email: frederique@ntu.edu.sg
} \and Gilles Zémor
\thanks{Institut de Math\'ematiques de Bordeaux, UMR 5251, universit\'e
  de Bordeaux, France}}
\begin{document}
\maketitle 

\begin{abstract}
We consider oblivious transfer protocols performed over binary symmetric
channels in a malicious setting where parties will actively cheat if they can. 
We provide constructions purely based on coding theory that achieve an 
explicit positive rate, the essential ingredient being the existence of linear codes
whose Schur products are asymptotically good.
\end{abstract}
%
%
\section{Introduction}

A 1-out-of-2 oblivious transfer is a cryptographic protocol between
two players, Alice, who owns two secrets, and Bob, who wishes to
acquire one of them. The protocol ensures that one of the secrets
is delivered to Bob, while no information about the other secret
leaks:
furthermore, Alice has no information about which secret Bob selects.

1-out-of-2 oblivious transfer protocols were introduced by Even,
Goldreich and Lempel \cite{EGL}, though it was shown that they are
equivalent to a variant originally proposed by
Rabin \cite{Rabin}. Oblivious transfer
has found numerous applications since, notably to multiparty
computation \cite{Harnik}. 
Oblivious transfer was first considered in the case
of computationally bounded participants, but Crépeau and Kilian
later \cite{CK} introduced the idea of unconditionally secure
oblivious transfer, by considering the situation when the players 
have access to a noisy channel that they cannot control.
It is natural to view the noisy channel as a resource, and to use
the number of symbols sent over the channel as measure of 
the efficiency of the oblivious transfer protocol.
The first protocol achieving oblivious transfer of $1$-bit secrets
from a binary
symmetric channel (BSC) \cite{CK} required $\Omega(n^{11})$ transmitted bits 
over the channel  to guarantee a probability of protocol failure of $2^{-n}$.
 This was improved by Cr\'epeau in \cite{Crepeau}, who obtained a
 protocol that requires only $O(n^3)$ uses of the BSC for a one bit
 secret, and was later again improved
 by Crépeau, Morozov and Wolf \cite{CMW} to $O(n^2+\varepsilon)$.

The notion of oblivious transfer capacity was introduced in \cite{NW}
and further developed in \cite{Ahlswede}.
In this setting, the two secrets are not single bits anymore, but strings
of bits, and the oblivious transfer rate is defined as the quotient of
the number of
bits of each secret divided by the number of bits transmitted over the
channel: the oblivious transfer capacity is then equal to the supremum of the
set of
achievable rates. Oblivious transfer is studied in \cite{Ahlswede} in the
semi-honest (or honest but curious) model, 
meaning that the players do not deviate from the protocol.
In this model it is shown that constant rate oblivious
transfer is possible for most discrete memoryless channels.
Capacities are computed for some channels in \cite{Ahlswede}, and
lower bounds provided
for others, including the binary symmetric channel, though an actual
protocol is not proposed.

Later, Ishai et al.~\cite{Ishai} improved the results of Cr\'epeau
\cite{Crepeau}, and 
are the first to show that constant rate oblivious transfer protocols
are possible in a fully malicious setting where both players will
actively cheat if they can.
The paper \cite{Ishai} also devises a protocol that achieves efficient
oblivious transfer of many $1$-bit secrets in parallel. The protocols
of \cite{Ishai} are quite intricate and call upon a number of
cryptographic primitives. In the present paper we again pick up the
issue of devising constant rate oblivious transfer protocols and apply
a coding-theory approach to the problem. The end result consists of
constant rate oblivious transfer protocols that are more direct than
that of \cite{Ishai} and that allow us to compute an achievable rate
that is of different order of magnitude than what could eventually be
derived from \cite{Ishai}.\footnote{Y. Ishai, personal communication.}
Crucial to our protocols is the notion of Schur product of a linear
code $C$. The Schur product (or square) of a linear code $C\subset \FF_q^n$ is the
linear span of all coordinate-wise products $c*c'=(c_1c_1'\cdots
c_nc_n')$ of codewords $c=(c_1,\ldots ,c_n), c'=(c'_1,\ldots ,c'_n)$
of $C$.
The central ingredient in our construction is a family of
asymptotically good linear codes whose squares are also asymptotically
good.
In the case of the binary alphabet, 
the existence of such families of codes is far from immediate and they were not
known to exist before the work of Randriambololona~\cite{Hugues}. 
In fact, during the early investigations leading up to this work,
we realised the usefulness of such codes but were unable to
come up with a construction, and inquiries into the matter provided
motivation for the paper \cite{Hugues}, as is mentioned in its
introduction.
Codes with good squares also appear indirectly in \cite{Ishai}, since
they are an essential component of the secret sharing schemes with
multiplicative properties that \cite{Ishai} calls upon. The use of
codes with good squares is arguably more direct in the present work.
In the next section we give an overview of our constructions and
outline the structure of the paper.

\section{Overview}\label{sec:overview}
We assume Alice and Bob have access to two channels: (1) a noiseless
channel and (2) a discrete memoryless channel.
How unconditional oblivious transfer can be achieved  is best understood in the simple case
of an erasure channel of erasure probability $p$, say. First Alice
generates a string $\rv$ of $2n_0$ random bits and sends them over the
noisy channel.
Bob will receive approximately $2pn_0$ erasures instead of the original
symbols.
Bob then separates the index set $[1,2n_0]$ of the received symbols into
two disjoint sets of equal size $n_0$, i.e. $[1,2n_0]=I\cup J$,
in such a way that all the erased symbols have their coordinates 
in one of the two sets
(assuming there are no more than $n_0$ of them, which is typically the
case when $p<1/2$) and communicates the sets $I$ and $J$
noiselessly to Bob. 
Denoting by $\rv_I$ and $\rv_J$ the corresponding two $n_0$-bit
strings derived from the bits initially sent by Alice, she can now use them to
send noiselessly to Bob $\xv + \rv_I$ and $\yv +\rv_J$ where $\xv$ and
$\yv$ are some $n_0$-bit vectors. With this procedure Alice 
sends to Bob the vectors $\xv$ and $\yv$ through what amounts to two different
channels, one of which is noiseless, the other being in effect an
erasure channel. Alice has no way of knowing which of the two channels
is the noiseless one, meaning she cannot know which secret Bob will
want, and from Bob's side, whatever may be the way he chooses the two sets
$I$ and $J$, at least one of the two vectors $\xv$ and $\yv$ will be
submitted to an erasure channel of erasure probability at least $p$.
We remark now that all that is needed to complete the protocol is to apply standard
wiretap-channel techniques to transmit messages through the vectors
$\xv$ and $\yv$ that leak no information to an eavesdropper that would
access $\xv$ and $\yv$ through a channel of erasure probability at
least $p$. 
The oblivious transfer capacity for this honest but curious setting is explicitly computed in \cite{Ahlswede} as a function of the error probability $p$.

We now focus on our central topic, 
namely the case when the noisy channel is a binary symmetric
channel. All known protocols start with the following bit duplication trick,
first introduced by Cr\'epeau and Kilian \cite{CK} and also used in
\cite{Crepeau}. Alice again
generates a string $\rv$ of $2n_0$ random bits, but this time every bit
$r_i$ is sent over the noisy channel as a duplicate couple
$(r_i,r_i)$. We remark that whenever a couple $(0,1)$ or $(1,0)$ is
received, then either $(0,0)$ or $(1,1)$ must have been sent with equal probability $1/2$ at the receiver (Bob's) end. Therefore Bob has no choice than to consider this situation as an erasure, and what duplication achieves is to transform the binary
symmetric channel into a mixed channel with errors and erasures.
Again, Bob partitions the index set $[1,2n_0]$ into two $n_0$-bit sets $I$
and $J$, one of which indexes all the erased positions. This again
creates two vectors $\rv_I$ and $\rv_J$, one of which is received with
more noise, namely a mixture of errors and erasures, than the other
which is erasure-free. We have effectively created two virtual noisy
channels one of which is noisier than the other, and such that Alice
does not know which is the noisiest.
At this point we make the remark that wire-tap channel techniques are
again sufficient to extract from these two channels a semi-honest
oblivious transfer protocol. We develop this approach in Section
\ref{sec:P0}, which requires a treatment of the somewhat non-standard
mixed error-erasure wiretap channels. The result is a constructive oblivious
transfer protocol $P_0$ for an $m$-bit secret that is a generalisation
of a protocol of \cite{Crepeau} for $1$-bit secrets, and that achieves the
lower bound on the oblivious transfer capacity computed in
\cite{Ahlswede} for a mixed error-erasure channel. The computations of 
\cite{Ahlswede} are purely information-theoretic and no explicit
schemes were suggested to achieve them. By optimising over the channel
parameter we obtain a positive rate $\Rc_0=0.108$ for this first
protocol $P_0$.
For the rate $\Rc$ of a 1-out-of-2 oblivious transfer protocol (two secrets)
we use the following definition, consistent with \cite{Ahlswede}:

\begin{defn}\label{def:rate}
The {\em rate} $\Rc$ of an oblivious transfer protocol of one out of
two $m$-bit secrets 
is the ratio of the number of secrets bits, namely $2m$, over the
number $N$ of binary symbols transmitted over the channel.
\end{defn}

Protocol $P_0$ ensures that an honest but curious Bob will have no knowledge on
at least one of the two secrets. To measure the possible leakage of information
about a secret, we first view the two secrets as uniform and independent random variables 
in $\{0,1\}^m$. We write them therefore as $X,Y$. The protocol $P_0$ would be ideal if we
could state regarding Bob's view that:
\begin{align*}
\text{Either}\quad & H(X|Y,\mathcal{O})=m\\
\text{Or}\quad     & H(Y|X,\mathcal{O})=m
\end{align*}
where $\mathcal{O}$ is what Bob observes during the protocol and $H$ is
Shannon's entropy. We prove a lower bound on $H(X|Y,\mathcal{O})$ that
explicitly states how close (in fractions of bits) the protocol is from the ideal scenario. 

Protocol $P_0$ is only valid in a semi-honest model where Alice
does not deviate from her instructions. Alice's goal in cheating is restricted
to trying to figure out the secret that Bob wants (she is not interested in
disrupting the protocol, i.e. to make it fail). 
Contrary to the pure erasure
channel case, Alice could actively cheat by transmitting over the
binary symmetric channel some falsely duplicated bits $r_i$ under the
form $(0,1)$ or $(1,0)$, instead of $(r_i,r_i)$. If $(0,1)$ is sent
over a binary symmetric channel with transition probability $p<1/2$,
then the probability $p^2+(1-p)^2$ that $(0,1)$ or $(1,0)$ is received (an erasure)
is always larger than if $(0,0)$ or $(1,1)$ had been transmitted.
By sending a few tracker pairs of symbols in this way, a tellingly
large number of their indices will end up in the subset, $I$ or
$J$ corresponding to the secret that Bob does not want, thus yielding 
critical information to Alice on which secret Bob is trying to acquire.

A crucial observation made by Cr\'epeau \cite{Crepeau} is that the
number of falsely duplicated tracker bits $r_i$ that Alice can use can only
be a limited portion of the total number of bits transmitted over the
noisy channel. This is because these bits have a higher probability of
turning up on Bob's side as erasures, and if Bob receives too many
erased symbols, contradicting the law of large numbers, he will know
that Alice has almost certainly cheated. Hence if one repeats the
protocol $P_0$ many times, say $n_0^2$ times where $n_0$ is (as above, up
to a multiplicative constant) the number of noisy channel uses for
$P_0$, this makes the number of channel uses equal to $n_0^3$, and the
number of corrupt tracker bits that Alice can get away with using
without arousing Bob's suspicion, is, by the law of large numbers,
 not significantly more than the order of
$n_0^{3/2}$: this implies that Alice has to
be honest for the majority of the $n_0^2$ $P_0$-protocols that are
played out, otherwise she will be exposed with probability tending to
$1$ with $n_0$.

The following idea is then used by Cr\'epeau \cite{Crepeau} to obtain
an oblivious transfer protocol secure against malicious participants
that would cheat if they could. The treatment of \cite{Crepeau} focuses
on oblivious transfer of single bit secrets, but it applies just as
well to string oblivious transfer. Apply $n$ times the protocol $P_0$
to intermediate secret pairs $x_i,y_j$, $i=1,\ldots,n$. The strings
$x_i\in\{0,1\}^m$ are chosen randomly such that $x_1+\cdots +x_n= s$
and the $y_i$ are defined as $y_i=x_i+s+t$, $i=1,\ldots,n$, where $s$ and $t$
are Alice's secrets to be obliviously transferred. Now to foil Alice's
tracking strategy, Bob will, for every $i$, randomly ask for either
$x_i$ or $y_i$, taking care only to ask for an even number of ${y_i}'s$
if he wishes to eventually acquire $s$, and an odd number of ${y_i}'s$
if he wishes to acquire $t$. We see that summing all the intermediate
secrets Bob has acquired yields either $s$ or $t$ according to his
wish, and Alice who can only cheat on a fraction of the $P_0$
protocols, obtains no information on the eventual secret, $s$ or $t$,
obtained by Bob.

The above repetition scheme gives vanishing rates however. The core
strategy developed in the present paper is to again repeat $n$ times the
protocol $P_0$, but to replace the condition
$x_1+\cdots +x_n= s$ by a generalised condition
\begin{equation}
  \label{eq:syndrome}
  \HH\begin{bmatrix}x_1\\ \vdots \\ x_n \end{bmatrix} = \sv
\end{equation}
where $\HH$ is a suitably chosen binary $r\times n$ matrix, yielding
secrets $\sv$ and $\tv$ of length $rm$ rather than $m$ (every coefficient of $\sv$ lives in $\{0,1\}^m$). It turns out that central
among the required properties of $\HH$ is that it generates a binary
linear code
with a square that has a large minimum distance.
We develop this approach in Section \ref{sec:posOT} where we propose first
an oblivious transfer protocol $P_1$ that prevents Alice from cheating
but introduces cheating possibilities for Bob (whose goal is to obtain information about the other secret than the one being asked), and then introduce a
variation $P_1'$ of $P_1$ which prevents both Alice and Bob from
cheating. The protocol $P_1'$ adds to the protocol $P_1$ a compression
function applied to $\sv$ and $\tv$.
These protocols achieve a positive rate of respectively
$\Rc_1\approx 0.69\; 10^{-4}$ and $\Rc_1'\approx 0.34\; 10^{-4}$.

In Section~\ref{sec:qary}, we introduce generalisations $P_2$ and
$P_2'$ of protocols $P_1$ and $P_1'$ that replace the binary code
generated by $\HH$ in \eqref{eq:syndrome} by a $q$-ary code, for $q$ a power of 2. This
allows us to use algebraic geometry codes with a much improved rate, 
with the drawback
that the protocol $P_0$ has to be replaced by a less efficient
1-out-of-$q$ semi-honest oblivious transfer protocol. Overall, the rates
of the protocols $P_2$ and $P_2'$ improve upon $P_1$ and $P_1'$,
giving $\Rc_2\approx 1/1250$, $\Rc_2'\approx 1/2500$.

We finish this overview with a formal definition of oblivious transfer considered in the malicious case.

\begin{defn}\label{def:OT}
Given a noiseless channel with unlimited usage and a binary symmetric channel (BSC), 
a {\em 1-out-of-2 oblivious transfer protocol of rate $\Rc$} consists of a two
player protocol where one player, Alice, possesses two secrets $\sv$ and $\tv$
of $m$ bits, and the second player, Bob, asks Alice for one of the two secrets.
The protocol uses communication over both channels and should satisfy the following properties:
\begin{enumerate}
\item
The ratio of the number $2m=|\sv|+|\tv|$ of secrets bits over the total number
$N$ of binary symbols transmitted over the noisy channel is $\Rc$.
\item
The protocol is correct, meaning that if Bob follows the protocol, he will obtain the secret he wishes with probability 
that tends to $1$ when $N$ goes to infinity.
\item
Bob, whether he cheats or not, has virtually no information on at least one
secret, meaning
\begin{align*}
\text{Either}\quad & H(S|T,\mathcal{O})\geq H(S)-\delta\\
\text{Or}\quad     & H(T|S,\mathcal{O})\geq H(T)-\delta
\end{align*}
where $S$ and $T$ are $\sv$ and $\tv$ viewed as random variables with uniform
distribution, $\mathcal{O}$ is what Bob observes during the protocol and $\delta$
is a quantity that tends to $0$ when $N$ goes to infinity.
\item
A cheating Alice who is trying to gain non-trivial information on which secret Bob is
asking for will either fail at obtaining anything or be accused of cheating by
Bob with probability tending to $1$ when $N$ goes to infinity.
It may happen that Bob accuses Alice of cheating when she is behaving honestly,
but this happens with probability that tends to $0$ when $N$ goes to infinity.
\end{enumerate}
\end{defn}

We remark that:
\begin{enumerate}[leftmargin=*]
\item[(1)]
In what follows, we will always assume that the secrets of Alice are two
independent uniformly distributed strings. This assumption is made without loss
of generality. Indeed,  since the noiseless channel is assumed to be available at no cost, Alice may always one-time pad her secrets $\sv$ and $\tv$ by computing $\sv+\xv$ and $\tv+\yv$ and communicating them noiselessly to Bob, for some independent uniform
random strings $\xv$ and $\yv$. After this, oblivious transfer of the
secrets $\sv$ and $\tv$ is equivalent to oblivious transfer of the random
strings $\xv$ and $\yv$.
\item[(2)]
This condition on secrecy used in \cite{Ahlswede} is
$H(K_{\bar{Z}}|Z,\mathcal{O})$, where $Z$ is a random variable that models the
choice of a secret, and $K_{\bar{Z}}$ represents the secret which was not
chosen. Our condition is slightly stronger since it assumes the complete knowledge of one secret is given. 
Deviating from the definition of \cite{Ahlswede} was required 
since it makes no sense in the malicious context to model the choice of a secret by a binary random variable. 
Indeed, we will see that in some instances Bob can try to extract from the
protocol some mixture of partial information from both secrets.
\end{enumerate}

Our main results are the protocols  $P_1,P_2$ and $P'_1,P_2'$. Protocols $P_1'$
and $P_2'$ satisfy Definition~\ref{def:OT} with all probabilities that are
required to tend to zero doing so subexponentially, i.e. scaling as
$\exp(-N^\alpha)$ for $0<\alpha<1$. The quantity $\delta$ in Point 3. of the
definition is also subexponential in $N$. Protocols $P_1$ and $P_2$ are preliminary versions of 
protocols $P_1'$ and $P_2'$ where only Alice is fully malicious while Bob is
assumed to be honest-but-curious.

%
%
\section{A First Binary Oblivious Transfer Protocol}
\label{sec:P0}

Alice and Bob have access to two channels: (1) a noiseless channel and
(2) a binary symmetric channel (BSC) with crossover probability
$\varphi <1/2$.
The binary field over $\{0,1\}$ is denoted by $\FF_2$, and for $a \in \FF_2$, $\bar{a}$ denotes the other element of $\FF_2$.

The protocol $P_0$ below is a slight variation of the oblivious
transfer proposed by Cr\'epeau in~\cite{Crepeau}, allowing Alice's two
secrets to be strings of
$m=n_0\epsilon(1-h(\tfrac{\varphi^2}{(1-\epsilon)}))$ bits, instead of 1
bit, where $h$ denotes the binary entropy function, and where we have set
$\epsilon = 2\varphi(1-\varphi)$.
The total number of noisy channel uses is $N=4n_0$.

\begin{framed}
{\bf Protocol $P_0$.} 
Alice has two (column) secrets $x,y\in\FF_2^m$. Alice and Bob agree on an
$((n_0-k)+ m)\times n_0$ binary matrix ${\bf H}'$ of the form  
\[
{\bf H}'=
\begin{bmatrix}
{\bf H}_0 \\
{\bf H}_1
\end{bmatrix}
\]
where ${\bf H}_0$ is the parity check matrix of some $(n_0,k)$ linear code $C_0$, which is capacity achieving over a BSC with crossover probability $\tfrac{\varphi^2}{(1-\epsilon)}$, and comes with an efficient decoding algorithm, while ${\bf H}_1$ is chosen uniformly at random. 
\begin{enumerate} 
\item
Alice generates a string ${\bf r}=(r_1,\ldots,r_{2n_0})$ of $2n_0$ random bits and sends $2n_0$ pairs $(r_i,r_i)$ of random bits to Bob over the BSC.
\item
For every pair of the form $(r_i,r_i)$ or $(\bar{r}_i,\bar{r}_i)$, Bob
decides that the bit $r_i$ or $\bar{r}_i$ is successfully received. He
declares an erasure if he receives $(r_i,\bar{r}_i)$ or
$(\bar{r}_i,r_i)$. Bob partitions the indices $[1,2n_0]$ into two sets:
$I$ has size $n_0$ and contains only indices corresponding to
successfully received bits, while $J$, also of size $n_0$, contains the
rest of the indices. This is assuming that Bob wants to know the
secret $x$: if instead he prefers the secret $y$, then he will reverse
the roles of $I$ and $J$. To each set corresponds a string of noisy
random bits ${\bf r}'_{I}$ and ${\bf r}'_{J}$. 
For ${\bf r}'_{I}$, the noise comes from Bob accepting
$(\bar{r}_i,\bar{r}_i)$ while Alice sent $(r_i,r_i)$. For ${\bf
  r}'_{J}$, the noise also includes erasures. Bob sends both sets of
indices $I$ and $J$ to Alice over the noiseless channel. Alice then permutes uniformly at random elements in $I$ and in $J$ and sends the permutations to Bob over the noiseless channel.
\item
Alice picks uniformly at random two (column) codewords ${\bf c}_x$ and
${\bf c}_y\in C_1$ that satisfy respectively
\[
{\bf H}'\cv_x=
\begin{bmatrix}
{\bf 0}\\
x
\end{bmatrix},~
{\bf H}'\cv_y=
\begin{bmatrix}
{\bf 0}\\
y
\end{bmatrix}
\]
and sends ${\bf c}_x+{\bf r}_{I}$ and ${\bf c}_y+{\bf r}_{J}$ to Bob over the noiseless channel. 
\item
Bob computes $({\bf c}_x+{\bf r}_{I})+{\bf r}'_{I}$ to find ${\bf
  c}'_x$, a noisy version of ${\bf c}_x$. Bob decodes ${\bf c}'_x$,
recovers ${\bf c}_x$, and computes ${\bf H}_1\cv_x=x$.  
\end{enumerate}
\end{framed}

The protocol $P_0$ requires $2n_0$ uses of the BSC channel for each secret. It provides an oblivious transfer protocol provided that Alice is honest.
 
\begin{paragraph}{Suppose both Alice and Bob are honest.}
If Alice is honest and sends pairs of the form $(r_i,r_i)$, Bob receives $(r_i,r_i)$ with probability $(1-\varphi)^2$ and $(\bar{r}_i,\bar{r}_i)$ with probability $\varphi^2$. 
He will decide an erasure with probability $2\varphi(1-\varphi)=\epsilon < 1/2$, and accept a bit with probability $1-\epsilon$. 
This can be seen as an instance of an imperfect binary erasure channel
(BEC) with erasure probability $\epsilon$: when Bob decides that a random bit is correctly received, there is still a probability $\varphi^2$ of getting the wrong random bit.
\begin{itemize}
\item 
Since Bob accepts a pair of random bits with probability $1-\epsilon$ ($\epsilon<1/2$), he should receive on average $2(1-\epsilon) n_0$ non-erased symbols, and $I$ can be assumed of size $n_0$.  
\item
The string ${\bf r}'_{I}$ is ${\bf r}_{I}$ affected by an additive noise, that is 
${\bf r}'_{I} = {\bf r}_{I}+{\bf e}$, where ${\bf e}$ is an error vector, which contains a 1 whenever Alice sent $(r_i,r_i)$ and Bob received $(\bar{r}_i,\bar{r}_i)$. If Alice is honest, a bit flip happens with probability $\varphi^2$, thus the proportion of bit flips among the bits that are not erased is $\tfrac{2n_0\varphi^2}{2n_0(1-\epsilon)}$, yielding, by restricting over the non-erased bits, a binary symmetric channel (BSC) with crossover probability $\tfrac{\varphi^2}{(1-\epsilon)}$.
It is enough that the chosen error capability of the code allows an honest Bob
to recover ${\bf c}_x$. However, by choosing a code which is capacity achieving,
$n_0$ is minimized, and $m$ is maximized, as shown below, while discussing the
optimization of $\Rc_0$, the rate of $P_0$. Polar codes \cite{Arikan} provide examples of capacity achieving codes for the BSC which furthermore come with an efficient decoding algorithm. 
\end{itemize}
\end{paragraph}

\begin{paragraph}{Suppose Bob is dishonest.}
We now check that Bob, even if he is dishonest and deviates
from the protocol by putting indices of erased positions in both sets
$I$ and $J$,
will not recover any information about at least one of the two
secrets, 
that is, he cannot recover information involving both ${\bf c}_x$ and ${\bf c}_y$.
Bob gets roughly {$2n_0(1-\epsilon)$} bits (the rest being erased), he can thus partition the $2n_0$ bits into two groups in any way he wants, where $I$ will have some bits erased, some not, and the rest will be in $J$. 
Bob will receive some ${\bf c}'_x$ and ${\bf c}'_y$, which are noisy
versions of ${\bf c}_x$ and ${\bf c}_y$ respectively, with noise
depending on the choice of $I$ and $J$. These noises can be seen as
the result of a transmission through a channel between Alice and Bob
that behaves as a mixture of an erasure channel and a binary symmetric
channel.
Concretely, the noisiest of the two channels will have an average proportion
of erased symbols that is at least $\epsilon$, and its non-erased symbols
 are all submitted to a binary symmetric channel of
transition probability $\varphi^2/(1-\epsilon)$, as they were before
the partition into $I$ and $J$, since Bob has no way of differentiating
symbols in error from error-free symbols.

For the purposes of the present study, let us call a {\em binary
  symmetric channel with erasures} of parameters $e\in [1,n_0]$ and
$0\leq p<1/2$, where
$e$ is an integer, a
channel which acts on strings of $n_0$ bits in the following way:
\begin{dashlist}
  \item
  it erases $e$ coordinates chosen uniformly among all possible $\binom{n_0}{e}$
  patterns (in the protocol $P_0$, Alice permutes the indices in $I$ and $J$ uniformly at random),
  \item
  it applies a binary symmetric channel of transition probability
  $p$ to the remaining $n_0-e$ symbols.
\end{dashlist}

We will rely on the following result:

\begin{lem}\label{lem:min-entropy}
Let $C$ be a binary code of length $n$ and rate $R$, and let $C_X$ be
a random variable with values in $C$ and uniform distribution.
Let $C_X$ be submitted to a binary symmetric channel with erasures of
parameters $e$ and $p$, and let $Z$ be the output variable. 
Define the conditional min-entropy of $C_X$
given $Z=z$  by $$H_{\infty}(C_X|Z=z)=-\log \max_{c\in C} P(C_X=c
|Z=z).$$ 
Then, for all $\alpha>0$,  with probability that tends to $1$ exponentially in $n$, a
vector $z$ is received such that
$$H_{\infty}(C_X|Z=z) \geq n[R-(1-e/n)(1-h(p))-\alpha].$$
\end{lem}
The proof of Lemma~\ref{lem:min-entropy} is given in the Appendix.

In the case under study, the value of $e$  may
vary, but the probability that $e/n_0$ falls significantly below
$\epsilon =2\varphi(1-\varphi)$, i.e. is separated from $\epsilon$ by
a constant, is exponentially small in $n_0$. To obtain a uniformly
distributed secret from the transmitted codeword $\cv_x$
or $\cv_y$ in the protocol $P_0$, it suffices to hash it to a
sufficiently smaller string, which is
exactly the purpose of the multiplication by $\HH_1$. Since the set of
multiplications by $\HH_1$ makes up a universal family of hash
functions, we will invoke the Leftover Hash Lemma \cite[Lemma 4.5.1]{HILL} \cite[Theorem 3]{Bennett} to evaluate how close the protocol $P_0$ is from the ideal scenario
\begin{align}
\text{Either}\quad & H(X|Y,\mathcal{O})=m \label{eq:idealX}\\
\text{Or}\quad     & H(Y|X,\mathcal{O})=m \label{eq:idealY}
\end{align}
where we view the two secrets $x$ and $y$ as uniform and independent random variables $X,Y$ in $\{0,1\}^m$ and $\mathcal{O}$ is what Bob observes during protocol $P_0$.
The nature of protocol $P_0$ is such that $H(X|Y,\mathcal{O})=H(X|\mathcal{O})$, because $X$ and $Y$ are really transmitted
over two independent channels. 
Without loss of generality we assume that $X$ is transmitted over the noisiest of the two channels.
We have the following lower bound on $H(X|\mathcal{O})$:

\begin{theorem}\label{theor:HBob}
Suppose protocol $P_0$ is implemented with some $(n_0,k)$ linear code $C_0$ of
rate $R_0=k/n_0$. For any $\varepsilon>0$, whenever the length $m$ of the secret
satisfies $m\leq
n_0[R_0-(1-\epsilon)(1-h(\tfrac{\varphi^2}{1-\epsilon}))-\varepsilon]$, then $H(X|\mathcal{O})\geq m-f_0(\varepsilon,m)$, for $f_0(\varepsilon,m)$ exponentially small in $m$.
\end{theorem}
\begin{proof}

What is observed by Bob is a noisy version $z$ of a codeword $c$ sent through a
binary symmetric channel with erasures of parameters $e$ and
$p=\varphi^2/(1-\epsilon)$, with $e/n_0$ arbitrarily close to $\epsilon$.
Lemma~\ref{lem:min-entropy} claims that with probability tending to~$1$
(exponentially in $n_0$, meaning with probability $1-\exp(-n_0)$), Bob observes
$\omega=z$ such that the min-entropy $H_\infty(c|\mathcal{O}=\omega)$ of the
transmitted codeword $c$ is at least $n_0[R_0-(1-e/n_0)(1-h(p))]-\alpha n_0$, with $\alpha$ arbitrarily small. We then invoke Theorem 3
of \cite{Bennett} and the fact that the Renyi entropy is never less than the
min-entropy to claim that, since $X=\bH_1c$, we have
\[
H(X|\mathcal{O}=\omega,\bH_1) = H(\bH_1c|Z=z,\bH_1) \geq m - 2^{m-n_0[R_0-(1-e/n_0)(1-h(p))]+\alpha n_0-\log_2\ln 2}.
\]
Since $m\leq n_0[R_0-(1-\epsilon)(1-h(\tfrac{\varphi^2}{1-\epsilon}))]-\varepsilon$,
then $m-n_0[R_0-(1-e/n_0)(1-h(p))]+\alpha n_0 -\log_2\ln 2 \leq -4\beta n_0$ for some
$\beta=\beta(\varepsilon)>0$, which shows that we are already close to $m$ in a
way which is exponential in $n_0$, given $\bH_1$. Next, we remove the dependency
on $\bH_1$. We just showed that on average over $\bH_1$,
\begin{equation}\label{eq:average}
H(X|\mathcal{O}=\omega,\bH_1) \geq m - 2^{-4\beta n_0}.
\end{equation}
Suppose now that we were to be unlucky and choose $\bH_1$ in the set of
$\mathcal{H}_1$ of ``bad'' matrices $h$ (that may depend on $\omega$)
such that $H(X|\mathcal{O}=\omega,\bH_1=h) \leq m - 2^{-2\beta n_0}$. Since
\begin{eqnarray*}
H(X|\mathcal{O}=\omega,\bH_1) 
&=& \sum_{h\in\mathcal{H}_1\cup \overline{\mathcal{H}}_1} P(\bH_1=h)H(X|\mathcal{O}=\omega,\bH_1=h) \\
&\leq & \sum_{h\in\mathcal{H}_1}P(\bH_1=h)(m - 2^{-2\beta n_0}) +
\sum_{h\in\overline{\mathcal{H}}_1}P(\bH_1=h)m \\
&=& - \sum_{h\in\mathcal{H}_1}P(\bH_1=h) 2^{-2\beta n_0} +
\sum_{h\in\mathcal{H}_1\cup\overline{\mathcal{H}}_1}P(\bH_1=h)m \\
&=& m - 2^{-2\beta n_0} \sum_{h\in\mathcal{H}_1}P(\bH_1=h),
\end{eqnarray*}
we upper bound the quantity $H(X|\mathcal{O}=\omega,\bH_1)$ by 
\[
H(X|\mathcal{O}=\omega,\bH_1) \leq m - 2^{-2\beta n_0} P_u
\]
where $P_u$ is the probability to choose $\bH_1=h$ in $\mathcal{H}_1$. Together
with \eqref{eq:average} the above inequality implies that $P_u \leq 2^{-2\beta n_0} $. Therefore with probability $1-P_u \geq 1-1/2^{2\beta n_0}$ over the choice of the random matrix $\bH_1$, we have
\[
H(X|\mathcal{O} = \omega) \geq m - 2^{-2\beta n_0}.
\]

Now Lemma \ref{lem:min-entropy} does not exclude the existence of a ``bad" event $\omega\in\Omega_1$, 
for which we cannot guarantee \eqref{eq:average}. But we can write
\begin{eqnarray*}
H(X|\mathcal{O}) 
&=& \sum_{\omega\in\Omega_1\cup\overline{\Omega}_1} P(\mathcal{O}=\omega)H(X|\mathcal{O}=\omega) \\
&\geq & \sum_{\omega\in\overline{\Omega}_1}P(\mathcal{O}=\omega)(m - 2^{-2\beta n_0}) \\
&=& (m - 2^{-2\beta n_0})(1-2^{-\gamma n_0})
\end{eqnarray*}
where $2^{-\gamma n_0}$ is the probability of a bad event $\omega$.

We make the final remark that the above computation assumed that $e/n_0$ is
arbitrarily close to $\epsilon$. Of course, the number of erasures can deviate
significantly from the average: but this happens with probability exponentially
small in $n_0$, so that again this rare event can only diminish
$H(X|\mathcal{O})$ by a quantity exponentially small in $n_0$.
\end{proof}

\begin{cor}\label{cor:C0}
If $C_0$ is capacity-achieving on the erasureless channel, meaning the
code $C_0$ has a vanishing decoding error probability for a BSC of parameter
$p=\varphi^2/(1-\epsilon)$ and a rate $R_0$ arbitrarily close to
$1-h(\tfrac{\varphi^2}{1-\epsilon})$, and if $m \leq
n_0\epsilon[1-h(\tfrac{\varphi^2}{1-\epsilon})-\varepsilon]$, then Bob can only
obtain a vanishingly small number of bits of information on one of the two secrets.
\end{cor}

\end{paragraph}

\begin{paragraph}{Optimization of the length $m$ and of the rate $R_0$.}
Finally, $m$ is maximized by maximizing $R_0$, that is by having Alice use a
capacity achieving code for the relevant BSC channel, for which 
Corollary~\ref{cor:C0} has just told us that we may set
\[
m=n_0\epsilon(1-h(\tfrac{\varphi^2}{1-\epsilon})-\varepsilon),
\]
for an arbitrarily small positive $\varepsilon$. This gives us
an oblivious transfer rate (see Definition~\ref{def:rate}) arbitrarily close to
\[
{
\Rc_0=\frac{m}{2n_0}=\frac{n_0 \epsilon(1-h(\tfrac{\varphi^2}{1-\epsilon}))}{2n_0}=
\varphi(1-\varphi)\left(1-h\left(\frac{\varphi^2}{1-2\varphi(1-\varphi)}\right)\right)}.
\]
This is the lower bound on the oblivious transfer capacity found by Alshwede and
Csisz\'ar in \cite[Example 1]{Ahlswede}. What this shows is therefore that 
we incur no penalty on the achievable oblivious transfer rate by assuming a possibly malicious Bob
as opposed to the honest but curious Bob of \cite{Ahlswede}.

The optimal value of $m$ is obtained when $\varphi \approx 0.198$, $\epsilon \approx 0.31$, for which $\epsilon(1-h(\tfrac{\varphi^2}{1-\epsilon}))\approx 0.216$ (see Figure \ref{fig:ploteps}) and
\[
\Rc_0\approx\frac{0.216 n_0}{2n_0}=0.108.
\]
\end{paragraph}

\begin{figure}
\centering
\includegraphics[scale=0.5]{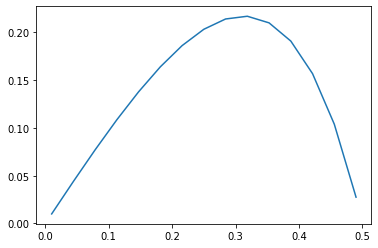}
\caption{\label{fig:ploteps} $\epsilon(1-h(\varphi^2/(1-\epsilon)))$ is shown as a function of the erasure probability $\epsilon$ of the imperfect BEC. 
}
\end{figure}

{\bf Suppose Alice is dishonest.} She might send a pair of the form
$(\bar{r}_i,r_i)$. Now Bob will receive $(\bar{r}_i,r_i)$ or
$(r_i,\bar{r}_i)$, that is an erasure, with probability
{$\varphi^2+(1-\varphi)^2=1-\epsilon > 1/2$}, in which case he will put the
index $i$ in the more ``noisy'' set $J$. Since it is more likely that a symbol of the form
$(r_i,\bar{r}_i)$ stays an erasure (rather than being changed into the 
valid symbol $(r_i,r_i)$ or $(\bar{r}_i,\bar{r}_i)$), by tracking which set contains the most indices on
which she cheated, she can guess which is more likely to be $I$ or
$J$.
For this reason, the protocol $P_0$ is only valid in a semi-honest
model where Alice is assumed not to deviate from the protocol.

To obtain a protocol valid against a malicious Alice that will
send falsely duplicated pairs of the form $(\bar{r}_i,r_i)$, we may
repeat $n=n_0^2$ times the protocol $P_0$, as in \cite{Crepeau} and as sketched in
Section~\ref{sec:overview}. 
If Alice is honest, the first time, she sends $(r_{1,1}r_{1,1}),\ldots,(r_{1,2n_0}r_{1,2n_0})$, the second time, she sends $(r_{2,1}r_{2,1}),\ldots,(r_{2,2n_0}r_{2,2n_0})$, $\ldots$. 
Let $Z$ be the random variable counting the number of valid
(i.e. non-erased) bits that Bob should receive.
It is binomially distributed, with mean {$E[Z]=2nn_0 (1-\epsilon)$}, and standard deviation $\sigma=\sqrt{2nn_0\epsilon(1-\epsilon)}$.

Suppose now Alice is dishonest, and that she cheats by sending $M$
falsely duplicated pairs out of the total $2nn_0$ transmitted pairs for all
the $n$ iterations of protocol $P_0$. Every time she cheats and
sends $(\bar{r}_{i,j}r_{i,j})$, Bob will declare an erasure with
probability {$1-\epsilon$}, therefore 
 the average number of valid bits that Bob will get is 
\begin{equation}\label{eq:E[Z]}
E[Z]= (2nn_0-M)(1-\epsilon) + M \epsilon = 2nn_0 (1-\epsilon) - M(1-2\epsilon),
\end{equation}
with $\epsilon<1/2$. If Alice cheats at least once on every one of the
$n$ instances of the protocol $P_0$ (so that $M\geq n$), then
for $n=n_0^2$, the typical value of $Z$ will
deviate from $2nn_0 (1-\epsilon)$ by a quantity that is much too close
to $n_0^2$ than the standard deviation of $Z$, that behaves as
$n_0^{3/2}$, allows. This tells Bob that Alice is cheating.

A bit more specifically, Bob will set a {\em threshold} $\tau$
to be equal to 
\begin{equation}\label{eq:threshold}
\tau = 2nn_0\left(1-\epsilon - \eta\right)
\quad
\text{with}
\quad \eta=\frac{1}{4n_0}(1-2\epsilon),
\end{equation}
which is exactly midway between the expected number $2nn_0(1-\epsilon)$
of unerased symbols he
should receive if Alice does not try to cheat and the expected number
\eqref{eq:E[Z]} of
unerased symbols he should receive if Alice cheats sufficiently many times
($M=n$) to access his secret. Bob will declare that Alice cheats if
the  number of unerased symbols that he receives falls below the threshold
$\tau$. The Chernov-Hoeffding inequality tells us that the probability that
Alice succeeds in cheating without being accused scales as $\exp(-\eta^22nn_0)$, and
similarly
that the probability that Bob wrongly accuses Alice of cheating
is also $\exp(-\eta^22nn_0)$.
With $n=n_0^2$, we have $\exp(-\eta^22nn_0)=\exp(-n_0)$.

More generally, we note that Alice's cheating will be almost surely
noticed whenever the number $M$ of corrupted bits she sends satisfies
\begin{equation}
  \label{eq:Acheat}
  M\gg \sqrt{nn_0}.
\end{equation}

Conversely, whenever the order of magnitude of $M$ stays below $\sqrt{nn_0}$, she
gets away with her behaviour.
We notice in particular that to prevent Alice from cheating on at
least one bit for every instance of $P_0$, the number $n$ of times
the protocol $P_0$ must be repeated has to satisfy $n\gg n_0$.

In the next section we further exploit this strategy of repeating $n$ times
$P_0$ to devise an efficient oblivious transfer protocol secure
against a cheating Alice.

%
%
\section{A Positive Rate Binary Oblivious Transfer Protocol}
\label{sec:posOT}
\subsection{A protocol that defeats Alice's cheating strategy}
\label{sec:defeatA}
We recall the definition of a Schur product of codes over the finite
field $\FF_q$, for $q$ a prime power. Schur products were possibly explicitly first used in Coding Theory for decoding applications \cite{Pellikaan} and later came
under attention in cryptographic contexts, in part because of their relevance to
secret sharing and multiparty computation. For details on 
applications see the introduction of \cite{CCMZ} and for a survey of their properties see \cite{Hugues2}.
\begin{defn}\label{def:schur}
Given a $q$-ary linear code $C$ of length $n$, for $q$ a prime power,
the Schur product (or square) of $C$, denoted $\hat{C}$, is defined as the linear span of
all componentwise products $c*c'$ of code vectors $c,c'$ of $C$, i.e.
$$\hat{C} = \langle c*c', c,c'\in C\rangle$$
with
\[
c*c'= (c_1c_1',\ldots,c_nc_n').
\]
\end{defn}
We will denote the length, the dimension, and the minimum Hamming distance 
of the code $C$ and the code $\hat{C}$ respectively by $[n,r,d]$ and $[n,\hat{r},\hat{d}]$. 

For the moment we restrict ourselves to $q=2$.
Let now ${\bf H}$ be an $r\times n$ binary matrix of rank $r$ whose $i$th row is denoted by $H_i$, so that
\[
{\bf H}=\begin{bmatrix}
H_1 \\
\vdots \\
H_r
\end{bmatrix}
\]
satisfying 
\begin{equation}\label{eq:delta}
H_i \transpose{H_j}  =\delta_{ij},~i,j=1,\ldots,n. 
\end{equation}

The protocol $P_1$ described below provides an oblivious transfer between Alice
and Bob, assuming this time that Bob is honest (but not Alice). We will from now
on use the letter $m$ to denote the secret size in protocol $P_0$: the secret
size in protocol $P_1$ will be equal to $rm$, for an integer $r$ equal to the
dimension of a binary linear code $C$ that we now introduce.

\begin{framed}
{\bf Protocol $P_1$.}
Alice and Bob agree on an $r\times n$ binary matrix ${\bf H}$
satisfying (\ref{eq:delta}), which forms the generator matrix of an
$[n,r,d]$ code $C$. The dimension $r$ of $C$ 
and the minimum distance $\hat{d}$ of the square $\hat{C}$ should both
be linear in $n$.

Alice has two secrets 
\[
{\sv}=\begin{bmatrix}{s}_1\\ \vdots\\ {s}_{r}\end{bmatrix},
~{\tv}=\begin{bmatrix}{t}_1\\\vdots\\{t}_{r}\end{bmatrix} 
\]
with coefficients $s_i,t_i$ in $\FF_2^{m}$.
\begin{enumerate}
\item
\label{item:y}
Alice then picks uniformly at random a vector
$\xv=[x_1,\ldots,x_n]$, $x_i\in\FF_2^{m}$, such that
\[
{\bf H}\transpose{\xv}={\sv}
\]
and computes 
\begin{equation}\label{eq:y=}
\yv = \xv+\sum_{j=1}^r (s_j+t_j)H_j=[y_1,\ldots,y_n],~y_i\in\FF_2^{m},
\end{equation}
such that ${\bf H}\transpose{\yv}={\tv}$.
\item
\label{item:u}
Bob computes a binary vector ${\bf u}=[u_1,\ldots,u_n]\in \FF_2^n$ which is orthogonal to $\hat{C}$.
\item\label{item:xy}
If Bob wants the secret $\sv$ (respectively $\tv$), then
for every coefficient $u_{\ell}$ of ${\bf u}$, Bob asks Alice through the
protocol $P_0$ for the string
\begin{itemize}
\item $x_\ell\in\FF_2^{m}$ (respectively $y_\ell$) if $u_\ell=0$,
\item $y_\ell\in\FF_2^{m}$ (respectively $x_\ell$) if $u_\ell=1$.
\end{itemize}
\item
\label{item:vl}
After $n$ rounds of the protocol $P_0$, Alice has sent Bob the
requested $n$-tuple $\vv=[v_1,\ldots,v_n]$ which may be expressed as
\begin{align}\label{eq:v}
\vv &= \xv*(1+\uv) + \yv*\uv \quad \text{if Bob requested $\sv$}\\ \label{eq:v2}
\vv &= \xv*\uv  + \yv*(1+\uv) \quad \text{if Bob requested $\tv$}
\end{align}
\item\label{item:Hv}
Once Bob gets $\vv$, he computes ${\bf H}\transpose{\vv}$ to recover $\sv$ (or $\tv$).
\end{enumerate}
\end{framed}

We first remark that in the simple case when the matrix $\HH$ is a
single row made up of the all-one vector, $\HH =[1,1,\ldots, 1]$, then
the protocol $P_1$ reduces to the string version of Cr\'epeau's
oblivious transfer protocol sketched in Section~\ref{sec:overview}.
We now check that in the general case, protocol $P_1$ does what is required
of it when the players do not try to deviate. We need to check that
Bob indeed recovers $\sv$ or $\tv$ and obtains no information on the
other secret.

\begin{paragraph}{Suppose both Alice and Bob are honest.}
\begin{itemize}
\item
As needed in Step \ref{item:y} of $P_1$, the vector $\yv=\xv+\sum_{j=1}^r (s_j+t_j)H_j=[y_1,\ldots,y_n]$, $y_i\in\FF_2^{m}$, satisfies 
\[
{\bf H}\transpose{\yv}=\tv.
\]
Indeed, from (\ref{eq:delta}), ${\bf H}\transpose{H_j}=\ev_j$, where
$\ev_j$ is the weight 1 column vector with~1 at the $j$th position, and  
\[
{\bf H}\transpose{\yv}= {\bf H}\transpose{\xv}+\sum_{j=1}^r (s_j+t_j){\bf H}\transpose{H_j}
                 = \sv + \sum_{j=1}^r (s_j+t_j)\ev_j
                 = \sv +(\sv+\tv) = \tv.
\]
\item
To show that Bob can recover $\sv$ or $\tv$ by computing ${\bf H}\transpose{\vv}$ in Step \ref{item:Hv} of $P_1$, we first remark that $\vv$ can be expressed as a Schur product
\[
\vv=\xv+{\bf u}*(\xv+\yv)=[x_1+u_1(x_1+y_1),\ldots,x_n+u_n(x_n+y_n)],
\]
if Bob wants $\sv$, or 
\[
\vv=\xv+({\bf u}+{\bf 1})*(\xv+\yv)=[x_1+(u_1+1)(x_1+y_1),\ldots,x_n+(u_n+1)(x_n+y_n)],
\]
if Bob wants $\tv$, according to \eqref{eq:v} and \eqref{eq:v2}. Now Bob gets $\vv$, and computes 
\[
{\bf H}\transpose{\vv}=
\sv+{\bf H}\transpose{({\bf u}*(\xv+\yv))}
=
\sv+{\bf H} \transpose{({\bf u}*(\sum_{j=1}^r (s_j+t_j)H_j))},
\]
since $\yv = \xv+\sum_{j=1}^r (s_j+t_j)H_j$, and
\[
{\bf H} \transpose{({\bf u}*(\sum_{j=1}^r (s_j+t_j)H_j))}
= {\bf H} \transpose{(\sum_{j=1}^r (s_j+t_j)({\bf u}*H_j))}\\
 = \sum_{j=1}^r (s_j+t_j){\bf H} \transpose{({\bf u}*H_j)}.
\]
Now the $i$th row of the vector ${\bf H} \transpose{({\bf u} * H_j)}$ is
\[
H_i\transpose{({\bf u}*H_j)} = \sum_{k=1}^n H_{ik}u_kH_{jk}= {\bf u} \transpose{(H_i* H_j)}.
\]
Since ${\bf u}$ is orthogonal to $\hat{C}$, ${\bf u} \transpose{(H_i* H_j)}=0$, and as desired,  Bob gets $\sv$. Now if ${\bf u}+{\bf 1}$ is used instead of ${\bf u}$, we have 
\[
{\bf H}\transpose{(({\bf u}+{\bf 1})*(\xv+\yv))} ={\bf H}\transpose{({\bf u}*(\xv+\yv)+\xv+\yv)} = {\bf H}\transpose{(\xv+\yv)} 
\] 
and Bob gets $\tv$.
\end{itemize}

Additionally, 
we remark that Bob has obtained $\sv$ (say) and $\vv_1\in(\FF_2^{m})^n$ 
given by \eqref{eq:v} and that protocol $P_0$ guarantees that he
essentially has no information on the coefficients of the other
$n$-tuple $\vv_2$ given by \eqref{eq:v2}. 
More precisely, the argument above shows that in an idealised version of protocol $P_0$, where \eqref{eq:idealX} and \eqref{eq:idealY} hold,  Bob has no information from $\vv_2$
in the sense that, given $\vv_1$ and $\sv$,  all
possible values for $\tv=\HH\transpose{\vv_2}$ are equally likely, in other words
Bob has no information on $\tv$.
Let us now prove that Bob has almost no knowledge on $\tv$, even when given
$\vv_1$ (which implies knowledge of $\sv$), and the actual output of protocol
$P_0$. Given $\vv_1$, \eqref{eq:y=} proves that, for every fixed vector
${\mathbf u}$ chosen by Bob, the vector $\vv_2$ lives in a code whose codewords
are in one-to-one correspondence with the values of $\tv$. 
Let us denote by $V_2$ the random variable equal to $v_2$ with distribution
conditioned by the knowledge of $\vv_1$, (which does not depend on the actual
value of $\vv_1$).
We have that
protocol $P_0$ transforms every coordinate $(V_2)_\ell$ of $V_2$, $\ell=1\ldots
n$, into an $n_0$-tuple $Z$ of $\{0,1,*\}^{n_0}$ (where $*$ denotes an erasure), in
a way that is memoryless and without feedback. In other words, the distribution
of $Z_\ell$ conditional on $(V_2)_\ell$ is the same as the distribution of
$Z_\ell$
conditional on $(V_2)_1,\ldots ,(V_2)_\ell$, and the distribution of
$(V_2)_\ell$ conditional on $(V_2)_1,\ldots (V_2)_{\ell-1}$ {\em and} $Z_1,\ldots
,Z_{\ell-1}$ is the same  as the distribution of $(V_2)_\ell$ conditional on 
$(V_2)_1,\ldots (V_2)_{\ell-1}$ alone. These properties are well-known (e.g.
\cite[Ch. 7]{CoverThomas}) to imply that
\[
I(V_2,Z)\leq \sum_{\ell=1}^nI((V_2)_\ell,Z_\ell).
\]
From Theorem~\ref{theor:HBob} we have $I((V_2)_\ell,Z_\ell)\leq
f_0(m,\varepsilon)$, from which we get, since $H(V_2)=H(T)=rm$,
\[
H(T|\mathcal{O})\geq H(V_2|\mathcal{O}) \geq rm-nf_0(m,\varepsilon),
\]
where $T$ is $\tv$ viewed as a random variable with uniform distribution
and $\mathcal{O}$ is Bob's view of the whole protocol. 
In other words, the amount of information leaked in the whole process is at most
$nf_0(m,\varepsilon)$.
We have thus proved the following:

\begin{cor}\label{cor:P1}
Suppose that protocol $P_1$ is implemented using an $[n,r,d]$ code $C$
satisfying the requirements, and used by an honest Alice who owns two secrets $\sv$ and $\tv$ of length $r$ (linear in $n$), i.e. in $(\FF_2^{m})^r$, and an honest Bob. Then 
\[
H(T|S,\mathcal{O}) \geq rm -nf_0(m,\varepsilon),
\]
where $\varepsilon$ and $f_0(m,\varepsilon)$ are as in
Theorem~\ref{theor:HBob} and in particular $f_0(m,\varepsilon)$ is exponentially small in $m$.
\end{cor}

\end{paragraph}

\begin{paragraph}{Suppose Alice is dishonest.} In Step \ref{item:xy}
  of $P_1$, Bob asks Alice for either $x_\ell$ or $y_\ell$ via 
  protocol $P_0$. If Alice is honest, she has no information on
  whether Bob is asking for $x_\ell$ or for $y_\ell$.
Now  we know that Alice can cheat in protocol $P_0$, and might guess
whether Bob is asking for $x_\ell$ or $y_\ell$: however, among the $n$
iterations of $P_0$, Alice can only cheat up to $M$ times while staying under
the radar, as long as (from (\ref{eq:Acheat}))
 $M$ stays below a linear function of $\sqrt{({\rm \#~channel~uses~for~}P_0)n}$.
Keeping the notation of Section \ref{sec:P0}, take $n$ linear in
$n_0^{2}$ (say), so that $M$ cannot exceed a quantity linear in
$\sqrt{n_0n}=n_0^{3/2}=n^{3/4}$.

Obtaining information on whether Bob asks for $x_\ell$ or for $y_\ell$
is equivalent to obtaining information on the $\ell$-th coefficient
$u_\ell$ of the vector $\uv=(u_1,\ldots ,u_n)$.
Now $\uv$ is randomly chosen in $\hat{C}^\perp$, therefore 
the $M\approx n^{3/4}$ coefficients seen by Alice will be distributed
uniformly at random in $\{0,1\}^M$ as long as $M$  is less or equal to
the dual minimum Hamming distance of $\hat{C}^\perp$ (see \cite{MS},
Ch.5. $\S$5. Theorem 8), which is the minimum Hamming distance of
$\hat{C}$ set to be linear in $n$, assuming for the moment that
such codes exist. Therefore,
Alice cannot differentiate with only $M$ values of $\uv$ whether Bob
is asking for $\sv$ of for $\tv$: Alice's cheating strategy is foiled.

To be more specific, Alice can only gain something from her cheating attempt if
she cheats on $M$ instances of $P_0$ with $M$ exceeding the Hamming distance of
$\hat{C}$, i.e. $M\geq cn$ for some constant $c$. As discussed at the end of
Section~\ref{sec:P0}, Bob will set a threshold $\tau$ as in
\eqref{eq:threshold}, with the value $\eta$ being adjusted to
$\eta = \tfrac{c}{4n_0}(1-2\epsilon)$, so that $\tau$ sits exactly between the
expected number of unerased symbols he should receive when Alice does not try to
cheat, and the expected number of symbols he will receive when Alice chooses
$M=cn$. Again, the probability that either Alice cheats successfully without
being caught and the probability that Bob wrongly accuses Alice of cheating 
both scale like $\exp(-n_0)$.

We are left to show that codes $C$ with all the required properties
exist. The code $C$ should have positive rate, i.e. its dimension $r$
should be a linear function of $n$, so that the oblivious transfer
protocol has positive rate. As we have just seen, the minimum distance $\hat{d}$
should be large enough. We note that the protocol would still work
with a code $C$ such that $\hat{d}$ is $o(n)$.
However, whatever the value of $n$ viewed as a function of $n_0$,
we will always need
$\hat{d}\gg n^{1/2}$ which exceeds what one obtains with
straightforward constructions. In 
 \cite{Hugues}, H. Randriambololona showed the existence of
 asymptotically good Schur codes, that is, with both dimension $r$ and
 product minimum Hamming distance $\hat{d}$ linear in $n$. These codes
 will therefore suit our purposes. To be complete, we just need to
 show that we may incorporate the extra requirement
 \eqref{eq:delta}. We do this below.
\end{paragraph}

\begin{paragraph}{Existence of a suitable code.}
We will show that punctured subcodes of the codes of \cite{Hugues}
satisfy all the requirements of protocol $P_1$. Recall that if
$C\subset \FF_q^n$ is a linear code, and if $I\subset\{1,\ldots, n\}$
is a subset of coordinate positions, then the punctured code $C^I$ on
the subset $I$ is the set of vectors of length $n-|I|$
$$\xv^I := (x_i)_{i\in\{1,\ldots, n\}\setminus I}$$
obtained from all codewords of $\xv =(x_1,\ldots ,x_n)$ of $C$.
If $d$ is the minimum distance of $C$, and if the number $|I|$ of
punctured positions is $<d$, then the dimension of the punctured code
$C^I$ equals the dimension of $C$, and the minimum distance of $C^I$ is
at least $d-|I|$. We rely on the following lemma, that we state in a
general $q$-ary case since we shall require it in non-binary form in
the next section. 
Let us say that the vectors $H_1,H_2,\ldots ,H_r$ in $\FF_q^n$ make up
an {\em orthonormal basis} of $C$ if they satisfy \eqref{eq:delta}.
We shall use the notation $\scalprod{A}{B}$ for the scalar product
$A\transpose{B}$ of vectors $A$ and $B$. 

\begin{lem}\label{lem:orthog}
Let $q$ be a power of $2$ and let
$\FF_q$ be the associated finite field. Let
$C\subset \FF_q^n$ be a linear code of dimension $r$ and minimum
distance $d>r$. Then there exists a subset $I\subset \{1,\ldots ,n\}$
of at most $r$ coordinate positions, such that puncturing $C$ on the
set $I$ yields a code of dimension $r$ that has an orthonormal basis.
\end{lem}
\begin{proof}
Note that the condition $d>r$ ensures that puncturing on at most $r$
positions does not decrease the code dimension.

We start with a systematic generating matrix of $C$: denoting by
 $B_1,\ldots ,B_r$ its rows, we have, for $j=1,2,\ldots ,r$, that the
 $j$-th coordinate $B_{ij}$ of $B_i$ equals $B_{ij}=\delta_{ij}$. 

Let $I_1=\emptyset$ if $\scalprod{B_1}{B_1}\neq 0$ and $I_1=\{1\}$ if 
$\scalprod{B_1}{B_1}= 0$. In both cases we therefore have 
$\scalprod{B_1^{I_1}}{B_1^{I_1}}\neq 0$, and 
because $q$ is a power of $2$
every element in $\FF_q$ is a square, and therefore
we have that a non-zero multiple of $B_1$,
that we name $H_1$, satisfies $\scalprod{H_1^{I_1}}{H_1^{I_1}}=1$.

Next, suppose by induction that we have found $\ell$ codewords
$H_1,\ldots , H_\ell$ of $C$, $1\leq \ell\leq r-1$ and a subset
$I\subset\{1,\ldots ,\ell\}$ such that
\begin{enumerate}
  \item
  for every $1\leq i,j\leq \ell$, $\scalprod{H_i^{I}}{H_j^{I}} =
  \delta_{ij}$
  \item
  for every $i=1,\ldots ,\ell$, for every $i+1\leq j\leq r$, the
  $j$-th coordinate of $H_i$ satisfies $H_{ij}=0$.
\end{enumerate}
We show that we can add a codeword $H_{\ell+1}$ to $H_1,\ldots
,H_\ell$ and possibly add coordinate $\ell+1$ to $I$, while keeping
properties 1 and 2 above satisfied. This will prove the Lemma by
induction.
To this end consider the linear combination
$$\Lambda = B_{\ell+1} + \sum_{i=1}^\ell\lambda_iH_i.$$
There clearly is a choice of $\lambda_1,\ldots ,\lambda_\ell\in\FF_q$
that makes $\Lambda^I$ orthogonal to $H_1^I,\ldots ,H_\ell^I$. 
If $\scalprod{\Lambda^I}{\Lambda^I}\neq 0$ leave $I$ unchanged,
otherwise adjoin the element $\ell+1$ to it. Property~2
ensures that the orthogonality relations
$\scalprod{\Lambda^I}{H_i^I}=0$, $i=1\ldots \ell$ are unchanged.
The required code vector is $H_{\ell+1}=\lambda\Lambda$ where $\lambda$
is chosen so that $\lambda^2\scalprod{\Lambda^I}{\Lambda^I}~=~1$.
\end{proof}

Now suppose the code $C$ has dimension $r$ and square distance
$\hat{d}>r$.
We always have $d\geq\hat{d}$ (consider $\xv *\xv$ where $\xv$ is a
minimum weight codeword of $C$) so that Lemma~\ref{lem:orthog} applies 
and we obtain a punctured code of $C$ that has an orthonormal basis,
that has dimension $r$ and square
distance at least $\hat{d}-r$, since it should be clear that
puncturing and taking the square yields the same code as taking the
square and then puncturing. In particular if we start from the codes of
\cite{Hugues} that are guaranteed to have $\hat{d}$ at least equal to
a linear function of $n$, we may first take a subcode to ensure a
dimension that stays linear in $n$ but satisfies $r<\hat{d}$, and then
puncturing will yield a code with square minimum distance that still
behaves as a linear function of $n$. Actual rates are computed at the
end of this section.
\end{paragraph}

Protocol $P_1$ works under the assumption that Bob is honest. 
Now Bob may cheat in Step \ref{item:u}, and ask for any mixture of
$x_\ell$ and $y_\ell$ of his choice, that may differ from \eqref{eq:v}
and \eqref{eq:v2}, in an attempt to obtain some mixture of the two
secrets, e.g. some bits of $\sv$ and some bits of $\tv$, or some
sums of the bits of $\sv$ and $\tv$. 
The modified protocol $P'_1$ described next makes sure he cannot do
anything of the kind.

\subsection{Defeating Bob's cheating strategies}\label{sec:rate}
The protocol $P_1'$ below simply adds a compression function to the secrets
$\sv$ and $\tv$ of protocol $P_1$. The compression function is
revealed to Bob only after protocol $P_1$ has been performed.

\begin{framed}
{\bf Protocol $P'_1$.}
Consider the setting of Protocol $P_1$.
Alice has two secrets 
\[
\tilde{\sv}=\begin{bmatrix}\tilde{s}_1\\ \vdots\\ \tilde{s}_{u}\end{bmatrix},
~ \tilde{\tv}=\begin{bmatrix}\tilde{t}_1\\\vdots\\\tilde{t}_{u}\end{bmatrix} 
\]
with coefficients in $\FF_2^m$, and $u=r(\frac 12-\delta)$.
\begin{enumerate}
\item
Alice picks uniformly at random two matrices, $M_s$ 
and $M_t$ both of dimension
$r(\tfrac{1}{2}-\delta) \times r$, with coefficients in $\FF_2$,
and two $r$-dimensional vectors $\sv$, $\tv\in(\FF_2^m)^r$ such that
\[
M_s\sv = \tilde{\sv},~M_t\tv = \tilde{\tv}.
\]
\item
Alice and Bob perform Protocol $P_1$ with $\sv$ and $\tv$ computed above, 
so that Bob gets either $\sv$ or $\tv$.
\item
Finally Alice sends Bob the matrices $M_s$ and $M_t$, and Bob computes 
\[
M_s\sv=\tilde{\sv}~(\mbox{or }M_t\tv=\tilde{\tv}) 
\]
to get the secret he wanted.
\end{enumerate}
\end{framed}

\begin{paragraph}{Suppose Bob is dishonest.} In Step \ref{item:xy} of $P_1$, Bob
asks Alice either $x_\ell$ or $y_\ell$ via the protocol $P_0$. He could cheat by
asking for some choice of $x_\ell$ and $y_\ell$ that does not correspond to \eqref{eq:v}
or \eqref{eq:v2}.
As a result, in Step \ref{item:Hv}, Bob will get some vector $\vv$ whose
components $v_\ell$ are either $x_\ell$ or $y_\ell$. Specifically, from
\eqref{eq:y=} we have that Bob gets exactly a vector
\[
\vv = \xv + \transpose{(\sv + \tv)}{\mathbf H}*\uv
\]
where $\uv$ is an arbitrary row-vector chosen by Bob, the column vector $\sv +
\tv$ is chosen by Alice independently of $\xv$ (since $\tv$ can be any quantity
independent of $\sv$), and ${\mathbf H}*\uv$ can be taken to be the matrix deduced from
${\mathbf H}$ by replacing its $\ell$th column by the zero column whenever $u_\ell=0$.
With this convention we have $((\sv +\tv){\mathbf H})*\uv=(\sv+\tv)({\mathbf H}*\uv)$.

Consider now ${\mathbf H}\transpose{\vv}$. We have
\begin{align*}
{\mathbf H}\transpose{\vv} &= {\mathbf H}\transpose{\xv} + {\bf H}({\mathbf H}*\uv)^\intercal (\sv+\tv)\\
&= \sv + V(\sv+\tv)\\
&= (\Id+V)\sv +V\tv
\end{align*}
where $V$ is the $r\times r$ matrix ${\bf H}({\mathbf H}*\uv)^\intercal$ over
$\FF_2$. We will not attempt to characterize the set of possible matrices $V$
that arise in this way and simply assume that $V$ can be any binary $r\times r$
matrix. 
Let us also write
\begin{equation}\label{eq:UV}
{\bf H}\transpose{\vv} = 
U \sv+V \tv
\end{equation}
and remark that knowledge of $\sv+\tv$ gives us
$(\sv+\tv)^\intercal{\mathbf H}=\xv+\yv$ and enables us to turn $\vv$ into its
complement vector, i.e. with coordinate $y_\ell$ for every $v_\ell=x_\ell$
and with coordinate $x_\ell$ for every $v_\ell=y_\ell$. Therefore, for any
fixed $\uv$, there is a bijection between the couples $(\xv,\sv+\tv)$ and
$(\vv,\sv+\tv)$,
both of which live in $(\FF_2^m)^{n+r}$. We also remark that $\sv+\tv$ and
$\sv+V(\sv+\tv)$ gives us $\sv$ and therefore $\tv$: therefore the map
$(\sv,\tv)\mapsto (\sv+\tv,U\sv+V\tv)$ is one-to-one. The conclusion is that
knowledge of $\vv$ gives us $U\sv+V\tv$, and no additional knowledge on
$\sv+\tv$, hence no additional knowledge on the couple $(\sv,\tv)$. Henceforth
we forget all properties of $(U,V)$ stemming from their particular structure, except for this last fact.

We now prove that: 
\begin{prop}\label{prop:rankV}
One of the following holds:
\begin{itemize}
\item either $\rank V \leq r/2$, in which case Bob has no
information on $\tilde{\tv}$, even when he is given $\tilde{\sv}$, meaning precisely that 
\[
H(\tilde{T}|\tilde{S},\mathcal{O}) \geq rm\left(\frac{1}{2}-\delta\right) -
f_1(r,\delta)-nf_0(\varepsilon,m)
\]
where $f_1(r,\delta)$ is exponentially small in $r$, and where $f_0(\varepsilon,m)$ is from Theorem~\ref{theor:HBob}.
\item
or $\rank V > r/2$, in which case Bob has no
information on $\tilde{\sv}$, even when he is given $\tilde{\tv}$, which means
\[
H(\tilde{S}|\tilde{T},\mathcal{O}) \geq rm\left(\frac{1}{2}-\delta\right) -
f_1(r,\delta)-nf_0(\varepsilon,m).
\]
\end{itemize}
\end{prop}

We will rely on the following classical lemma (e.g. \cite[Ch. 14, exercice
8]{MS}):
\begin{lem}
\label{lem:abc}
Let $B,C$ be $a\times b$ and $a\times c$ binary matrices respectively, and let $A=[B|C]$ be the $a\times(b+c)$ matrix that is obtained from concatenating $B$ and $C$. Suppose $0\leq b<a<b+c$. Let $B$ be a fixed matrix of rank $b$, and let $C$ be chosen randomly and uniformly among all binary $a\times c$ matrices. Then $P(\rank{A}<a)\leq 1/2^{b+c-a}$.
\end{lem}

\begin{proof}[Proof of Proposition~\ref{prop:rankV}]
\hspace{1cm}\vspace{-4mm}
\begin{enumerate}[leftmargin=*]
\item $\rank V \leq r/2$. Suppose that Bob knows $\sv$ which is stronger than Bob knowing $\tilde{\sv}$, that is
\[ 
H(\tilde{T}|\tilde{S},\mathcal{O}) \geq H(\tilde{T}|S,\mathcal{O}).
\]
We first suppose that Bob's observation $\mathcal{O}$  reduces to $\vv$, 
i.e. the collection of
$n$ binary $m$-tuples $v_\ell$ that are either $x_\ell$ or $y_\ell$, whichever he has 
requested when executing the $\ell$-th instance of protocol $P_0$.
Then between $U\sv+V\tv$ and $\sv$, Bob obtains the fixed quantity $V\tv$ and
for him
$\tv$ is uniformly distributed among vectors of the form ${\bm \tau} +(\Ker
V)^m$, where we have identified $V$ with a linear map $\FF_2^r\to \FF_2^r$ and
where
${\bm \tau}$ is any fixed preimage of $V\tv$. The secret
$\tilde{\tv}$ may therefore be
any quantity in $M_t({\bm \tau} +(\Ker V)^m)$. Let $\mathcal{M}$ be the set of matrices $M_t$
such that $M_t(\Ker V)$ is the full image space $\FF_2^{r(\frac 12-\delta)}$.
Since we have supposed $\rank V \leq r/2$, we have $\dim\Ker V\geq r/2$, and there must exist at least $r/2$ linearly independent vectors in $\Ker V$. Consider the images by $M_t$ of these $r/2$ vectors: they make up the columns of a
uniform random $r(1/2-\delta)\times r/2$ matrix, which by Lemma~\ref{lem:abc},
is of full-rank $r(1/2-\delta)$
with probability at least $1-2^{-r\delta}$. In other words we have
$P(M_t\in\mathcal{M})\geq 1-2^{-r\delta}$. We therefore have: 
\begin{eqnarray}
H(\tilde{T}|S,\mathcal{O})
& = & H(M_t T|S,\mathcal{O}) \nonumber\\ 
& =& P(M_t\in\mathcal{M})H(\tilde{T}|S,\mathcal{O},M_t\in\mathcal{M})\nonumber\\
  &&+
P(M_t\not\in\mathcal{M})H(\tilde{T}|S,\mathcal{O},M_t\not\in\mathcal{M})\nonumber \\
& \geq & (1-\tfrac{1}{2^{r\delta}})H(\tilde{T}|S,\mathcal{O},M_t\in\mathcal{M}) 
         =(1-\tfrac{1}{2^{r\delta}})mr(\tfrac 12-\delta)\nonumber\\
& \geq & m(\tfrac{r}{2}-\delta) -\tfrac{mr}{2^{r\delta}}\nonumber\\
H(\tilde{T}|S,\mathcal{O})&\geq&
m(\tfrac{r}{2}-\delta) -f_1(r,\delta).
\label{eq:H(T|S,O)}
\end{eqnarray}

\item $\rank V > r/2$. Bob again obtains ${\mathbf z}=U\sv+V\tv$, and we suppose this time that he is given $\tilde{\tv}=M_t\tv$. Our goal is to show that Bob obtains no information
on $\tilde{\sv}=M_s\sv$.
First consider that the possible values of $\tv$ given $\tilde{\tv}$ are
${\bm\tau}+(\Ker M_t)^m$,
for some fixed ${\bm\tau}$ such that $M_t{\bm\tau}=\tilde{\tv}$. The possible
values of $V\tv$ are $V{\bm\tau}+(V\Ker M_t)^m$. We have $\dim (V\Ker M_t) =
\dim\Ker M_t-\dim(\Ker V\cap\Ker M_t)$. 

Now
the kernel $\Ker M_t$ is a random subspace of $\FF_2^r$ of dimension at least
$r(\tfrac 12+\delta)$. Since we have supposed $\rank V > r/2$ we have $\dim\Ker V<r/2$. Choose a basis of $\Ker V$, and add arbitrary vectors of $\FF_2^r$ so as to obtain a basis $B$ of some subspace $\langle B\rangle$ of $\FF_2^r$ of dimension $r/2$ that contains $\Ker V$. 
Applying Lemma~\ref{lem:abc} we obtain that $\langle B\rangle +\Ker M_t$ is of
full rank $r$ with probability at least $1-1/2^{r\delta}$. 

In this case we have $\dim(\langle B\rangle\cap\Ker M_t)=\dim(\Ker M_t)+\dim \langle B\rangle-r=\dim(\Ker M_t)-r/2$. Since $\Ker V\subset\langle B\rangle$ we also have
$\dim(\Ker V\cap\Ker M_t)\leq\dim(\Ker M_t)-r/2$. Therefore,
\[
\dim (V\Ker M_t)=\dim\Ker M_t-\dim(\Ker V\cap\Ker M_t)\geq r/2.
\]

Now we have that the set $\{\sv, U\sv+V\tv={\mathbf z}\}$ is an $\FF_2^m$
expansion of a translate of an
$\FF_2$-vector space of dimension at least $\dim(V\Ker M_t)$.
As before, the image under the random matrix $M_s$ of a fixed subspace
of dimension at least $r/2$ has maximum dimension $r(1/2-\delta)$ with probability at least
$1-2^{-r\delta}$. So with probability $(1-2^{-r\delta})^2$ both random
matrices $M_t$ and $M_s$ behave as desired and $\tilde{\sv}$ can be any vector 
in $(\FF_2^m)^{r(\frac 12-\delta)}$ with uniform probability. Therefore
similarly to \eqref{eq:H(T|S,O)} we obtain
\begin{equation}
H(\tilde{S}|T,\mathcal{O}) \geq m(\tfrac{r}{2}-\delta)
-f_1(r,\delta).\label{eq:H(S|T,O)}
\end{equation}
\end{enumerate}
The estimates \eqref{eq:H(T|S,O)} and \eqref{eq:H(S|T,O)} have been obtained
with the assumption that Bob's observation reduces to $\vv$. In the actual
protocol, for every execution of protocol~$P_0$, Bob obtains $v_\ell$, which is equal to
one of the two secrets $x_\ell$ or $y_\ell$, plus $f_0(m,\varepsilon)$ bits of
information on the other secret, as guaranteed by Theorem~\ref{theor:HBob}.
By the same argument as that preceding Corollary~\ref{cor:P1}, consisting of
viewing all the individual instances of Protocol $P_0$ as
the successive instantiations of
a discrete memoryless channel without feedback, we have that
Bob obtains at most $nf_0(m,\varepsilon)$ additional bits of
information, hence the expressions in Proposition~\ref{prop:rankV}.
\end{proof}

\end{paragraph}

\begin{paragraph}{Rate of the oblivious transfer protocol $P_1'$.}
The protocol $P_1'$ is instantiated with two secrets of length $u=r(\tfrac 12
-\delta)$, where 
$r$ is the length of the secrets in the protocol $P_1$ and $\delta$ is a
positive number that can be taken arbitrarily close to $0$. Thus, the rate of
$P_1'$ can be arbitrarily close to
\[
\Rc_1'=\frac{\Rc_1}{2}.
\]
Now the protocol $P_1$ requires $n$ uses of the protocol $P_0$, where 
$n$ is the length of the $[n,r,d]$ code $C$. The total length of a
secret is $r$ times the length of a secret of $P_0$. Therefore the
overall rate of the protocol $P_1$ is
$$\Rc_1=R\Rc_0$$
where $\Rc_0$ is the rate of $P_0$ and $R=r/n$ is the rate of the code $C$.
From Section~\ref{sec:P0} we have that $\Rc_0$ can be taken arbitrarily close to
the limiting value $\Rc_0=0.108$, and from \cite{Hugues}
we have an infinite family of linear codes of length $n$ and square
minimum distance $\hat{d}\geq n/1575$ and dimension $>\hat{d}$. From
the discussion at the end of Section~\ref{sec:defeatA} we get that a
punctured version of the codes of \cite{Hugues} with a rate $R$
arbitrarily close to $1/1575$ will satisfy all the conditions of
protocol $P_1$. We get therefore the achievable rate:
$$\Rc_1=\frac{0.108}{1575} \approx 0.69\; 10^{-4}$$
and hence $\Rc_1'\approx 0.34\; 10^{-4}$.
\end{paragraph}

\subsection{Summary and Comments}\label{sec:summary}
Let $N$ denote the total number of bits sent over the noisy channel during
protocol $P_1'$.
We have $N=4n_0n$, where the number of bits sent by Alice over the
noisy channel is $4n_0$ for each instance of protocol $P_0$,
and $n$ is the number of times protocol $P_0$ is repeated.
We have set $n=n_0^2$ to be specific, but this is somewhat arbitrary,
and any value $n=n_0^\alpha$, $\alpha>1$ would yield similar asymptotic
guarantees.
These guarantees are the following:
\begin{enumerate}[leftmargin=*]
\item When Bob follows the protocol precisely, he obtains the secret he wishes
for with probability at least $1-\exp(-N^\alpha)$ for some $0<\alpha<1$. This is
obtained by
using polar codes in protocol $P_0$ that achieve the capacity of the least noisy
(in effect erasureless) channel. We recall that this family of codes is
constructive and can be decoded in quasi-linear time, with a probability of a
decoding error that is guaranteed to be subexponential in the blocklength, i.e.
$\exp(-n_0^\beta)$, see \cite[Theorem 1]{TalVardy}. The probability of a decoding
error on at least one of the instances of $P_0$ scales therefore as
$\exp(-N^\alpha)$.
\item Theorem~\ref{theor:HBob} and Proposition~\ref{prop:rankV}
guarantee that whatever Bob does, on at least one of the two secrets he obtains
at most
a vanishing number of bits of information, that behaves like $\exp(-cm^\alpha)$
where $m$ is the secret size, $0<\alpha<1$, and $c$ is a constant that is
determined by how close we are to the limiting rate $R_1'$ computed in the last
subsection. In other words, $c$ is dependent on the values of $\varepsilon$ and
$\delta$ in Theorem~\ref{theor:HBob} and Proposition~\ref{prop:rankV}.
\item
The protocol is protected against a cheating Alice in the following sense: Bob
will declare Alice a cheater if he receives a total number of erased symbols
that exceeds a certain threshold $\tau$. If Alice cheats so as to
uncover which of the two secrets Bob is trying to obtain, she will
almost surely be accused, i.e. with probability $1-\exp(-n_0)=1-\exp(-N^{1/3})$.
It may happen that an honest Alice will be wrongly accused of cheating by Bob,
but this happens with a vanishingly small probability that scales as
$\exp(-n_0)=\exp(-N^{1/3})$.
\end{enumerate}

In the next section we improve upon the rate $R_1'$ computed in Section~\ref{sec:rate} by modifying protocol
$P_1$, so as to allow us to replace the binary code $C$ by a $q$-ary
one: codes with large square distances are easier to construct in the
$q$-ary case and better rates are obtained. The techniques require to measure
leakage of information to Bob and the probabilities of Alice successfully
cheating without being found out, or of Alice being wrongly accused of cheating
are unchanged, and we will compute the new limiting rates without explicitly
mentioning that we need to be $\varepsilon$ and $\delta$ away from them as
in Theorem~\ref{theor:HBob} and Proposition~\ref{prop:rankV}.
%
%
\section{A Positive Rate $q$-ary Oblivious Transfer Protocol}
\label{sec:qary}

Let $q$ be a power of $2$. Protocol $P_2$ below is a $q$-ary variant
of protocol $P_1$: it relies upon a version of protocol $P_0$ that is
a 1-out-of-$q$ oblivious transfer protocol, that, like $P_0$ does not
protect against Alice's cheating strategy. Let us denote $P_0^q$ this
protocol: we shall show later in this section how to transform the
original protocol $P_0$ into its $q$-ary version $P_0^q$. 
The matrix $\HH$ is this time a matrix over the field on $q$ elements,
it is the generating matrix of a code $C$, with otherwise the same
requirements as in the binary case, namely that the rows of $\HH$ make
up an orthonormal basis of $C$ and the square distance $\hat{d}$ of
$C$ grows linearly in $n$. The integer $m$ is now the length of the $q$ secrets 
in protocol $P_0^q$.
\begin{framed}
{\bf Protocol $P_2$.}
Consider the same setting as Protocol $P_1$, over $\FF_q$ instead of $\FF_2$, 
that is the $r\times n$ generating matrix ${\bf H}$ of the code $C$
has coefficients in $\FF_q$, and the two secrets $\sv$ and $\tv$ have
length $r$, with coefficients in $\FF_q^m$. 
\begin{enumerate}
\item
\label{step:vq}
As in protocol $P_1$, Alice picks uniformly at random a vector
$\xv\in(\FF_q^m)^n$ such that 
\[
{\HH}\transpose{\xv}={\sv}.
\]
She then computes
the $q-1$ vectors $\yv_i$, $i=2,\ldots,q-1$, with coefficients in $\FF_q^m$, given by
\[
\yv_i = \xv+\lambda_i\sum_{j=1}^r (s_j+t_j)H_j,
\]
where $\lambda_i$ runs through every non-zero element of $\FF_q$. We
set $\lambda_1=1$, so that
$\yv_1$ coincides with $\yv$ in protocol $P_1$.  We write
$\yv_i=[y_{i1},y_{i2},\ldots ,y_{in}]$.
\item
Bob computes a $q$-ary vector $\uv=[u_1,\ldots ,u_n]$ which is
orthogonal to $\hat{C}$.
\item
If Bob wants the secret $\sv$, then he asks Alice through the protocol
$P_0^q$ for the vector $\xv + \uv*(\xv+\yv_1)$. Equivalently, whenever
$u_\ell =\lambda_i$, Bob asks for the string $y_{i\ell}$.
If Bob wants the secret $\tv$
instead, he asks for the vector $\yv_1 + \uv*(\xv+\yv_1)$.
\item
\label{step:1q}
After $n$ rounds of the protocol $P_0^q$, Alice has sent Bob the
requested $n$-tuple $\vv=[v_1,\ldots,v_n]$ which has again exactly the
expression given by \eqref{eq:v} and \eqref{eq:v2}.
\item
Once Bob gets $\vv$, he computes ${\bf H}\transpose{\vv}$ to recover $\sv$ or $\tv$.
\end{enumerate}
\end{framed}

\begin{framed}
  {\bf Protocol $P_2'$.}

The protocol is is obtained from protocol $P_2$
in exactly the same way as protocol $P_1'$ is obtained from $P_1$, 
with the two secrets 
$\tilde{\sv}$ and $\tilde{\tv}$ of length $u=r(\tfrac 12 -\delta)$ having their coefficients in $\FF_q^m$.
\end{framed}

We first argue that Bob will indeed recover the secrets by computing 
${\bf H}\transpose{\vv}$. The proof from Protocol $P_1$ carries
through quite straightforwardly. As previously, we have that
$\xv+\yv_1$ belongs to the code $C$, and since $\uv$ is orthogonal to
$\hat{C}$, we have that the scalar product of a row of $\HH$ with
$\uv*(\xv+\yv_1)$ is equal to the scalar product of $\uv$ with the
$*$-product of two vectors of $C$, hence equals zero.
Therefore ${\bf H}\transpose{\vv}$ is always equal to either 
$\HH\transpose{\xv}$ or $\HH\transpose{\yv_1}$. 
We have $\HH\transpose{\xv}=\sv$ by choice of $\xv$ and
$\HH\transpose{\yv_1}=\tv$ through the orthonormal property of the rows
$H_i$ of $\HH$.

\begin{paragraph}{A 1-out-of-$q$ oblivious transfer $P_0^q$.}
Next, a 1-out-of-$q$ oblivious transfer protocol is needed for Bob to obtain every $v_i$, 
$i=1,\ldots,n$. It may be obtained by $q-1$ applications of the
1-out-of-2 oblivious transfer protocol $P_0$, as in \cite{BCR}.
If $x_1,x_2,\ldots ,x_q$ are the secrets to be transferred, Alice
chooses $q-1$ random strings $r_1,r_2,\ldots ,r_{q-1}$ uniformly among
all strings such that $r_1+r_2+\cdots r_{q-1}=x_q$.
Then protocol $P_0$ is applied to the $q-1$ pairs of secrets
$$(x_1,r_1), (x_2+r_1,r_2), \ldots ,(x_i+r_1+\cdots
r_{i-1},r_i),\ldots (x_{q-1}+r_1+\cdots +r_{q-2}, r_{q-1}).$$
We see that if Bob asks for the first term $x_i+r_1+\cdots r_{i-1}$ of
the $i$-th pair, he loses all chance of obtaining $r_i$, and the
subsequent $x_j$, $j>i$. He can obtain $x_i$ by asking for the second
term, $r_j$, of the preceding pairs for $j=1\ldots i-1$.
\end{paragraph}

\begin{paragraph}{Construction of the required code $C$.} 
We need a code with an orthonormal basis and a large square distance $\hat{d}$.
From Lemma~\ref{lem:orthog} and the discussion just afterwards, such a
code will be obtained as soon as we have a code of dimension $k=r$ and large
square distance $\hat{d}$ 
 with $k<\hat{d}$. We also want this code to have the largest possible
 dimension, so we try to obtain the code with the largest possible
 square distance $\hat{d}$ satisfying $k\geq\hat{d}$, from which we
 will then take a subcode so as to have $k<\hat{d}$.

 We turn to algebraic geometry codes (see e.g.
\cite{AGcodes,Stichtenoth}). Consider codes $C(D,G)$ of length $n$ over $\FF_q$, defined to be the image of the linear evaluation map $ev:L(G)\rightarrow \FF_q^n$, $f \mapsto ev(f)=(f(P_1),\ldots,f(P_n))$, where $D=P_1+\ldots+P_n$ is a divisor on an algebraic curve $\mathcal{X}$ for the rational points $P_1,\ldots,P_n$, $\FF(\mathcal{X})$ is the function field of the curve $\mathcal{X}$, 
and $L(G)=\{ f \in \FF(\mathcal{X})^*,~\sum_{P\in \mathcal{X}}\nu_P(f)P + G \geq 0\} \cup \{0\}$, 
and $G$ some other divisor whose support is disjoint from $D$.

As observed in \cite[Lemma 14]{Hugues2},
we have
\begin{equation}
  \label{eq:C(D,2G)}
  \hat{C}(D,G)= C(D,G)*C(D,G) \subset C(D,2G).
\end{equation}
Indeed, for $c,c'\in C(D,G),$ we have
$$c*c'=ev(f)*ev(f')=(f(P_1)f'(P_1),\ldots,f(P_n)f'(P_n))=ev(ff'),$$
and $ff'\in\FF(\mathcal{X})$ with 
\[
\sum_{P\in \mathcal{X}}\nu_P(ff')P + 2G = (\sum_{P\in \mathcal{X}}\nu_P(f)P + G)+(\sum_{P\in \mathcal{X}}\nu_P(f')P + G) \geq 0
\]
showing that $ff'\in L(2G)$.

Now the parameters of the evaluation code $C(D,G)$ are known to
satisfy, when the degree $\deg{G}$ of the divisor $G$ is strictly less
than $n$, \cite[Cor. II.2.3]{Stichtenoth}
$$d\geq n-\deg G\quad \text{and} \quad k \geq \deg G + 1 -g$$
where $g$ is the genus of the algebraic curve. From
\eqref{eq:C(D,2G)} we therefore also have,
as long as $2\deg G<n$,
$$\hat{d}\geq n -2\deg G.$$
To have $k=\hat{d}$ we shall therefore aim for a divisor $G$ of degree
satisfying
\begin{equation}
  \label{eq:degG}
  n = 3\deg G + 1-g.
\end{equation}
\end{paragraph}

\begin{paragraph}{Rate of the oblivious transfer Protocols $P_2'$.}
Let us first establish the rate $R$ of the code $C$. We will obtain
the best result for $q=16$. The Tsfasman-Vladut-Zink bound 
\cite{AGcodes,Stichtenoth} tells us that we may choose curves with
genus $g$ such that $n\rightarrow\infty$ and
$n/g$ is arbitrarily close to $\sqrt{q}-1=3$. Choosing $\deg G$ as in 
\eqref{eq:degG} gives us $\frac 1n\deg G\rightarrow 4/9$ and $k\geq
\deg G +1-g$ gives us a rate at least $1/9$ for the code. The actual
code used in the protocol is possibly a punctured version, but since
the rate can only increase by puncturing up to the minimum distance,
we may guarantee a rate arbitrarily close to $R=1/9$ for the code $C$.

Now the overall rate of the protocol $P_2$ is 
\begin{align*}
  \Rc_2 &= \frac{2rm}{n\text{$\#$ channel uses for $P_0^q$}}\\
      &= \frac{2rm}{n(q-1)\text{$\#$ channel uses for $P_0$}}\\
      &= \frac{1}{q-1}R\Rc_0
\end{align*}
where $\Rc_0$ is the rate of protocol $P_0$.
Hence,
  $$\Rc_2= \frac{1}{9\times 15}0.108=0.8\; 10^{-3}.$$

The rate $\Rc_2'$ of the protocol $P_2'$ is $\Rc_2'=\Rc_2/2$, hence
\[
\Rc_2'= 0.4\; 10^{-3}.
\]

\end{paragraph}

\section{Concluding comments}
Binary codes $C$ with large rate and such that $\hat{C}$ has a large minimum
distance would of course yield improved rates for protocol $P_1'$. How large can
these rates be is a very intriguing question.

Apart from exhibiting codes $C$ with some extraordinary $\hat{C}$ behaviour,
improving upon the rates of this paper probably involves some alternative
approaches to the problem, or other ways of using the potential of $q$-ary
codes.

As discussed in Section~\ref{sec:summary}, the probabilities of something going
wrong
(Bob does not get his secret, Alice cheats, Bob falsely
accuses Alice of cheating) are subexponential in the total number $N$ of
transmitted bits. An interesting avenue of research would be to find explicit
achievable rates that guarantee an exponential behaviour for these failure
probabilities.

\appendix

\section*{Appendix: Proof of Lemma \ref{lem:min-entropy}}
\label{app:proof}

Recall that the channel randomly introduces $e$ erasures and a number
of errors on the non-erased symbols. Let us first modify slightly the
channel by assuming a fixed number $w$ of errors, chosen uniformly
among the $\binom{n-e}{w}$ possible choices, instead of binomially
distributed errors.

Build a bipartite graph, consisting of the $2^{nR}$ codewords as
vertices on the left,  and on the right all possible ternary strings
over the alphabet $\{0,1,*\}$ with $e$ erasures
(an erasure is represented by the $*$ symbol).
thus the right hand side of the graph has
\[
{n \choose e}2^{n-e}
\]
vertices.
We put an edge connecting a codeword, i.e. a left vertex, to a ternary string if 
the ternary string can be obtained from the codeword by erasing $e$ coefficients, and flipping $w$ others. 
There are thus $\binom{n}{e}\binom{n-e}{w}$ outgoing edges from
every codeword node, and a total of 
\begin{equation}
  \label{eq:edges}
  2^{nR}{n \choose e}{n-e \choose w}
\end{equation}
edges connecting the left 
and right sides of the bipartite graph.
Define $r$ such that
\begin{equation}
  \label{eq:r}
  {n\choose e}2^{n-e}2^r = 2^{nR}{n \choose e}{n-e \choose w}
\end{equation}
and assume parameters have been chosen such that $r$ is positive.
 
Now Alice picks a codeword $c$ uniformly at random and sends it over
the channel: one error pattern of weight $w$ and $e$ erasures will
happen, all of them are equally likely. This means that after
transmission an edge of the graph has been chosen uniformly among the
total number \eqref{eq:edges} of edges. We now argue that most edges
are connected to a right vertex with large degree.

Let $\alpha >0$. 
The number of edges connected to binary strings on the right of the graph, whose degree is smaller than $2^{r-\alpha}$ is 
$N_1+2N_2+3N_3+\ldots+2^{r-\alpha}N_{2^{r-\alpha}}$ where $N_i$ counts
the number of right nodes whose degree is $i$, and $\sum N_i={n
  \choose e}2^{n-e}$. 
Since $i\leq 2^{r-\alpha}$ for every $i$,
$N_1+2N_2+3N_3+\ldots+2^{r-\epsilon}N_{2^{r-\alpha}}\leq
2^{r-\alpha}\sum N_i=2^{r-\alpha}{n \choose e}2^{n-e}$. The number of
edges connected to right nodes whose degree is bigger than
$2^{r-\alpha}$ is then bounded from below by
\[
{n \choose e}\left(2^{n-e}2^r-2^{r-\alpha}2^{n-e}\right)={n \choose e}2^{n-e}2^r(1-2^{-\alpha}).
\]
This shows that for any $0<\alpha\leq r$, with probability at least $1-1/2^{\alpha}$, Bob receives a vector $v$ such that, for any codeword $c$, 
$$P(C_X=c\, |\, Z=z) \leq \frac{1}{2^{r-\alpha}}$$
where $C_X$ is the input to the channel, with uniform distribution on
the code $C$, and $Z$ is the random variable consisting of the
received vector.

If the number $w$ of errors equals the expected number of errors
for $n-e$ transmitted bits over a BSC, i.e. $w=p(n-e)$, then we get
from \eqref{eq:r} 

$$r= R - (1-e/n)(1-h(p)) + o(1).$$

This proves a version of Lemma~\ref{lem:min-entropy} for a constant,
rather than binomially distributed number of errors. For a binomially
distributed number of errors, i.e. the result of an actual binary
symmetric channel, we may proceed as above by making the bipartite
graph weighted. We put an edge for every possible number of errors
$w$, $0\leq w\leq n-e$, and associate to every such edge the weight
$p^w(1-p)^{n-e-w}$. Concentration of measure around the mean number
$p(n-e)$ of errors ensures the same behaviour as in the constant error
case: we leave out the cumbersome details.


%
%
\section*{Acknowledgments}

The research of F. Oggier for this work was supported by the Singapore
National Research Foundation under Research Grant NRF-RF2009-07.
Part of this work was discussed while G. Z\'emor was visiting the Division of Mathematical Sciences 
in Nanyang Technological University, and while F. Oggier was visiting the Institute for Mathematics, Bordeaux University. The authors warmly thank both host institutions for their hospitality. They also wish to thank Yuval
Ishai for fruitful discussions and encouraging them to pursue this work.

%
%


\begin{thebibliography}{99}
\bibitem{Ahlswede}
R. Ahlswede, I. Csiszar, ``On Oblivious Transfer Capacity", in the
proceedings of {\em IEEE International Symposium on Information
  Theory}, Nice, 2007.

\bibitem{Arikan}
E. Arikan, ``Channel Polarization: A Method for Constructing Capacity-Achieving Codes for Symmetric Binary-Input Memoryless Channels'', {\em IEEE Transactions on Information Theory}, vol. 55, no. 7, July 2009. 
%
\bibitem{Bennett}
C.H. Bennett, G. Brassard, C. Cr\'epeau, U.M. Maurer, ``Generalized Privacy Amplification", {\em IEEE Trans. on 
Information Theory}, vol. 41, no. 6, November 1995.

\bibitem{BCR}
 G. Brassard, C. Cr\'epeau, J-M. Robert,
``Information theoretic reductions among disclosure problems,''
{\em 27th Symposium on Foundations of Computer Science}, 1986.
%
%

\bibitem{CCMZ}
Cascudo, R. Cramer, D. Mirandola, and G. Z\'emor,
``Squares of random linear codes'',
{\it IEEE   Trans. on Information
Theory}, IT-61, No 3 (2015) pp. 1159--1173.

\bibitem{CoverThomas}
T. M. Cover and J. A. Thomas, ``Elements of Information Theory,'' Wiley, 1991,
2006.

\bibitem{Crepeau}
C. Cr\'epeau, ``Efficient Cryptographic Protocols based on Noisy
Channels", {\em EUROCRYPT} 1997.

\bibitem{CK}
C. Cr\'epeau, J. Kilian, ``Achieving Oblivious Transfer using Weakened Security Assumptions,'' {\em 29th Symposium on Foundations of Computer Science}, 1988.
%


\bibitem{CMW}
C. Cr\'epeau, K. Morozov, S. Wolf,
``Efficient Unconditional Oblivious Transfer from Almost Any Noisy
Channel,''
{\em Security in Communication Networks}
LNCS Volume 3352, 2005, pp 47-59.

\bibitem{EGL}
S. Even, O. Goldreich, A. Lempel,
``A randomized protocol for signing contracts,''
{\em Communications of the ACM}
Vol. 28 N. 6, June 1985,
pp. 637-647.

\bibitem{Harnik}
D. Harnik, Y. Ishai, E. Kushilevitz, ``How Many Oblivious Transfers are Needed for Secure Multiparty Computation", {\em  Advances in Cryptology - CRYPTO 2007}.

\bibitem{HILL}
J. H{\aa}stad, R. Impagliazzo, L. A. Levin and M. Luby, ``A Pseudorandom
Generator from any One-way Function,'' {\it SIAM Journal on
  Computing}, Vol. 28 n. 4, pp. 1364-1396, 1999.

\bibitem{AGcodes}
T. H{\o}holdt, J. H. van Lint, R. Pellikaan, ``Algebraic geometry codes", in the 
{\em Handbook of Coding Theory}, 1998.
%

\bibitem{Ishai}
Y. Ishai, E. Kushilevitz, R. Ostrovsky, M. Prabhakaran, A. Sahai, J. Wullschleger, ``Constant-Rate Oblivious Transfer from Noisy Channels'', CRYPTO 2011:667-684.
%
\bibitem{MS}
F. Mc Williams, N.J.A. Sloane, ``The Theory of Error-Correcting Codes", {\em North-Holland Publisching Company}, 1977.
%
\bibitem{NaorPinkas}
M. Naor, B. Pinkas, ``Oblivious Transfer and Polynomial Evaluation", 
{\em Proceedings of 31 annual ACM Symposium on Theory of Computing (STOC '99)}, 1999. 
%
%
\bibitem{NW}
A. Nascimento, A. Winter, ``On the Oblivious Transfer Capacity of
Noisy Correlations'', Proc. ISIT 2006, Seattle, pp.1871-1875, 2006.

\bibitem{Pellikaan}
R. Pellikaan, ``On decoding by error location and dependent sets of error positions,''
{\it Discrete Math.} 106/107 (1992) pp. 369--381.
%
\bibitem{Rabin}
M. Rabin, ``How to Exchange Secrets by Oblivious Transfer,'' Tech. Memo TR-81, Aiken Computation Laboratory, Harvard University, 1981.
%

\bibitem{Hugues}
H. Randriambololona, ``Asymptotically good binary linear codes with
asymptotically good self-intersection spans," 
{\em IEEE Trans. on Information Theory}, vol. 59, no. 5, May 2013,
pp. 3038--3045.
%


\bibitem{Hugues2}
H. Randriambololona, ``On products and powers of linear codes
under componentwise multiplication'', in {\it Algorithmic Arithmetic,
  Geometry, and Coding Theory},  vol. 637 of Contemporary Math., AMS, 2015.
%

\bibitem{Stichtenoth}
H. Stichtenoth,
{\em Algebraic function fields and codes},
Springer, 1993.

\bibitem{TalVardy}
I. Tal, A. Vardy, ``How to Construct Polar Codes,''
{\em IEEE Trans. on Information Theory}, vol. 59, no. 10, 2013, pp. 6562--6582.
\end{thebibliography}
\end{document}